\documentclass[letterpaper,12pt, leqno]{article} 
\usepackage{amsmath,amssymb,amsfonts, amsthm} 
\usepackage[mathscr]{eucal}
\usepackage[round]{natbib}

\newtheorem{theorem}{Theorem}[section]
\newtheorem{proposition}{Proposition}[section]
\newtheorem{corollary}{Corollary}[section] 
\newtheorem{lemma}{Lemma}[section]
\newtheorem{definition}{Definition}[section]
\newtheorem{assumption}{Assumption}[section]
\theoremstyle{definition}
\newtheorem{remark}{Remark}[section] 
\newtheorem{example}{Example}[section]
\numberwithin{equation}{section}

\newcommand{\envspace}{\vspace{2mm}}

\newcommand{\mm}{market maker} 
\newcommand{\Real}{\mathcal{R}}

\newcommand{\PP}{\mathbb{P}} 

\newcommand{\Ptilde}{\widetilde{\PP}} 
\newcommand{\Stilde}{\widetilde{S}}

\newcommand{\Mtilde}{\widetilde{M}} 

\newcommand{\Btilde}{\widetilde{B}} 
\newcommand{\FF}{\mathcal{F}}
\newcommand{\HH}{\mathcal{H}}
\newcommand{\EE}{\mathbb{E}} 
\newcommand{\Etilde}{\widetilde{\EE}}

\newcommand{\covtilde}{\widetilde{\operatorname{Cov}}}
\newcommand{\half}{\frac{1}{2}} 
\newcommand{\ind}{{\mathbf 1}}
\newcommand{\Dm}{\mathbf{D}}

\newcommand{\variance}[1]{\langle #1 \rangle}

\newcommand{\ds}{\displaystyle}

\newcommand{\sigt}{\widetilde{\sigma}}

\def\keywordname{{\bf Key words:}} 
\newcommand{\keywords}[1]{\par\addvspace\baselineskip\noindent\keywordname\enspace\ignorespaces #1}

\title{Hedging in an equilibrium-based model for a large investor. \thanks{The author is grateful to Dmitry Kramkov for helpful discussions and insightful comments.}}
 
\author{David German \\ \vspace{-2mm} {\small Claremont McKenna College} 
  \\ \vspace{-2mm}{\small Department of Mathematical Sciences} 
  \\ \vspace{-2mm}{\small 850 Columbia ave} 
  \\ \vspace{-2mm}{\small Claremont, CA 91711, USA} 
  \\{\small \texttt{dgerman@cmc.edu}} 
}

\date{\today}

\begin{document}

\maketitle
  
\begin{abstract}
  We study a financial model with a non-trivial price impact effect. In this model we consider the interaction of a large investor trading in an illiquid security, and a market maker who is quoting prices for this security. We assume that the market maker quotes the prices such that by taking the other side of the investor's demand, the market maker will arrive at maturity with maximal expected wealth. Within this model we concentrate on the issue of contingent claims hedging.
\end{abstract}


\newpage
\section{Introduction}
\label{sec:intro}

The question of valuation of contingent claims for a small economic agent is
well studied in various settings. In the case of complete markets the price of
a contingent claim is the initial capital of the replication strategy (a
unique arbitrage-free price). For incomplete markets exact replication is
rarely possible. In this case the utility-based valuation approach described in
the previous section is often used.  The basic economic assumption (imposed
either implicitly or explicitly) behind the general incomplete model in
Mathematical Finance is:
\begin{center}
  ``The agent can trade any security in the \emph{desired} quantity at the
  \emph{same} price''.
\end{center}
The interpretation of this assumption is that the actions of the agent do not
affect prices of securities and that there is no shortage of any security in
any quantity.

One way to relax this assumption is to introduce the notion of liquidity into
the model.  Liquidity is a complex concept standing for \emph{the ease of
  trading a security}. (Il)liquidity can have different sources, such as
inventory risk -- \cite{Stoll:78}, transaction costs -- \cite{CvitKar:95}, uncertain
holding horizons -- \cite{Huang:03}, asymmetry of information -- \cite{GarPed:04},
demand pressure -- \cite{GarPedPot:06}, search friction -- \cite{DufGarPed:05},
stochastic supply curve -- \cite{CetinJarrowProtter:04} and demand for immediacy -- 
\cite{GrossMiller:88}, among many others (see \cite{AmihudMendPed:05} for a
thorough literature overview).

In this paper we will relax the small economic agent assumption by
considering a model where agent's actions move prices. In other words, we
shall study a financial model with a non-trivial \emph{price impact} effect.
A practical example of such a market is provided, for instance, by an
over-the-counter market for an illiquid security, where a market maker quotes
prices on demand. In practice it turns out that the price quoted depends on
the transaction size. To distinguish our case from the classical one we shall
refer to the economic agent trading on such a market as a ``large'' investor.

We will consider the interaction of a large investor trading in an illiquid
security, and a market maker who is quoting prices for this security. We will
assume that the market maker quotes the prices such that by taking the other
side of the investor's demand, she will arrive at maturity with maximal
expected wealth.  This idea was used in a recent
paper by \cite{GarPedPot:06} for the
discrete time case, but only when the utility function of the market maker is of an
exponential form. Using equilibrium-based arguments the authors of that
paper considered the question of the evaluation of contingent claims. However,
they did not study the question of hedging.

The novelty of our study is that we look at the problem of replication of
contingent claims in the model with price impact. Moreover, many of our
results are derived in the continuous time framework and with utility
functions of rather general form. We will show the existence of a unique pricing rule
for a broad class of derivative securities and utility functions, as well as
the existence of a unique trading strategy that leads to a perfect
replication.

Let us point out that our approach to the model of a large investor follows the
traditional framework of Economic Theory. We begin with economic primitives
(such as agent's preferences and market equilibrium) and {\it then} derive the
model. This is different from several papers in Mathematical Finance where the
nature of illiquidity is postulated {\it a priori}, see for example,
\cite{CvitMa:96}, \cite{CetinJarrowProtter:04}, \cite{BankBaum:04} and
\cite{Frey:98}.

The idea that the price is determined by the zero net supply condition on the
market with multiple agents that are solving their individual optimization
problems (maximizing their terminal utility) is not new and was studied in the
classical paper by Karatzas, Lehoczky and Shreve in
\cite{KarLehShreve:90}. Unlike \cite{KarLehShreve:90} where multiple agents
(\emph{small investors/liquid market}) with different utility functions are
considered, we consider only one representative agent (\emph{large
  investor/illiquid market}). This allows us to avoid complicated fixed point
arguments used in \cite{KarLehShreve:90}.

This paper is organized as follows. Section \ref{defs} defines the
basic concepts. Section \ref{existence_and_uniqueness_rep_strat}
defines the replication strategy that is suitable for our model and
discusses existence and uniqueness of such a strategy. Within this
section the assumption of completeness of the market with respect to
the price process is playing an important role. Section
\ref{sec:compl-with-resp} is devoted to the study of market
completeness. This section also contains some particular
examples. For instance, here we prove that in the framework of the
Bachelier model and under the assumption of an exponential utility for
the market maker we can replicate any convex (in appropriate sense) European-type contingent claim
(e.g. convex combinations of long calls).

\section{Large investor market model}\label{defs}

We assume that the uncertainty and the flow of information are modeled by a
filtered probability space $(\Omega,\FF, (\FF_t)_{0\le t\le T}, \PP)$, where
the filtration $\FF$ is generated by a $J$-dimensional Brownian Motion $B$,
that is,
\begin{equation}
  \label{eq:5}
  \FF_t = \FF_t^B, \quad 0\leq t\leq T.  
\end{equation}
Here $T$ is a finite time horizon, and $\FF=\FF_T$.

The security market consists of $J$ risky assets and a riskless asset. These
assets are traded between the investor and the market maker. We work in
discounted terms and (without loss of generality) assume that the return on
the riskless asset is zero. We denote by $\FF_T$-measurable random variables
$f=(f^j)_{1\le j\le J}$ the payoffs of the risky assets at maturity and by
$S^H=(S^H_t)_{0\le t\le T}$ the ($J$-dimensional) price process of the risky
assets under the condition that the investor is using the ($J$-dimensional)
trading strategy or \emph{demand process} $H = (H_t)_{0\le t\le T}$. Of
course, at maturity the price does not depend on the strategy:
\begin{displaymath}
  S^H_T = f, \text{ for all } H. 
\end{displaymath}

From here on we will implicitly understand that we have $J$-dimensional
processes, and without loss of generality we will use one-dimensional
notation. However, in Section \ref{sec:compl-with-resp} we will explicitly
point out the vectors and the matrices that appear in the proofs.

The market maker can be viewed then as a {\it liquidity provider}. She takes
the other side of the investor's demand, which can be positive, as well as
negative. We assume that the market maker always responds to the investor's
demand, that is the market maker always quotes the price (which turns out to
be a function of the trade size). The reason for this assumption is that the
market maker is naturally forced to quote the prices to achieve the
equilibrium by meeting the investor's demand. Of course by equilibrium we mean
that both parties are ``happy'' with the current prices, have no desire to act
to change these prices, and the supply is equal to the demand. In order to
describe ``happiness'' of the market maker we use the standard apparatus of
utility functions. We assume that the market maker has a utility function
$U:\Real\rightarrow\Real$, which is strictly increasing, strictly concave,
continuously differentiable, and satisfies the Inada conditions
\begin{align*}
  U'(-\infty)&=\lim_{x\rightarrow-\infty }U'(x)=\infty, \\
  U'(\infty)&=\lim_{x\rightarrow\infty}U'(x)=0.
\end{align*}

We shall also require the following two technical assumptions.
\begin{assumption}\label{as:exp_moments}
  The terminal value of the traded asset $\left.f\!=\!(f^j)_{1\le j\le J}\!\in\FF_T\right.$, and
  the terminal value of the contingent claim $g\in\FF_T$ have finite
  exponential moments, that is
  \begin{align*}
    \EE[\exp(\langle q,f\rangle)]<\infty,\quad \EE[\exp(rg)]<\infty, \quad
    q\in\Real^J, \quad r\in\Real.
  \end{align*}
\end{assumption}

\begin{assumption}\label{as:exp_utility}
  Utility function $U:\Real\longrightarrow\Real$ satisfies
  \begin{align}\label{eq:exp_utility}
    c_1 < -\frac{U'(x)}{U''(x)} < c_2 \text{ for some } c_1,c_2>0.
  \end{align}
\end{assumption}

\noindent Clearly, a linear combination of exponential functions of the form
\begin{align*}
  U(x)=\sum_{i=1}^N-c_i\frac{e^{-\gamma_ix}}{\gamma_i}, \quad \gamma_i,c_i>0,\
  x\in\Real
\end{align*}
satisfies the assumption above.

Notice that Assumption \ref{as:exp_utility} implies the Inada conditions.

We assume that the investor reveals his market orders (his demand process) $H$
to the market maker. The market maker responds to the investor's demand by
quoting the price, and by taking the other side of the demand. That is, if $H$
is the investor's strategy, then $-H$ is the market maker's strategy. In other
words, the market maker responds to the demand so that the market rests in
equilibrium (supply equals demand). The market maker is quoting the price in
such a way that she arrives at maturity with maximal expected wealth. Formally
this can be stated as 

\begin{definition}
  \label{def:price_proc}
  Let $x\in\Real$ be the initial cash endowment of the \mm. Let $f=(f^j)_{1\le
    j\le J}$ be an $\FF_T$-measurable contingent claim. Let $\left.H=(H^j)_{1\le
    j\le J}\right.$ be a predictable process. The equivalent probability measure
  $\PP^H\sim\PP$ is called \emph{the pricing measure of $f$ under demand $H$},
  and the semimartingale $S^H$ is called \emph{the price process of $f$
    under demand $H$} if
  \begin{align}
    \label{eq:def:density}
    \frac{d\PP^H}{d\PP}\triangleq \frac{U'(x-\int_0^TH_udS^H_u)}{\EE[U'(x -
      \int_0^TH_udS^H_u)]},
  \end{align}
  and the price process $S^H$ with the integral $\int HdS^H$ are martingales
  under $\mathbb{P}^H$. In particular,
  \begin{align*}
    S^H_t\triangleq\EE^H[f|\FF_t], \quad 0\le t\le T.
  \end{align*}
\end{definition}

Notice that the Definition \ref{def:price_proc} is rather general, as it does
not specify any conditions on the utility function $U$, the demand process
$H$, and the contingent claim $f$.

The above definition displays an intimate relationship between the price
process and the pricing measure. It may not be clear from the formulation of
Definition \ref{def:price_proc} that it reflects the mechanics of the market
described in the previous paragraph. However, notice that the density of
$\PP^H$ is chosen in such a way that the process $-H$ is indeed a solution to
the \mm's optimization problem (which will be defined below.) Naturally, the
semimartingale $S^H$ is defined in such a way that it is a martingale
under the pricing measure. It will become evident from the following lemma,
that the numerator of \eqref{eq:def:density} is nothing else but the \mm's
marginal utility.

\begin{lemma}\label{lemma:equivalent_def}
  Let $x\in\Real$ be the initial cash endowment of the \mm. Suppose
  $f$ satisfies Assumption \ref{as:exp_moments}, and $U$ satisfies
  Assumption \ref{as:exp_utility}. Let $H=(H^j)_{1\le j\le J}$ be a
  predictable process. Suppose that $S^H$ is the price process of $f$
  under demand $H$. Then $-H$ is the unique solution of the
  optimization problem
  \begin{align}\label{eq:u(x)}
    u(x)\triangleq\max_{G\in\HH(S^H, \PP^H)}\EE[U(x+\int_0^TG_udS^H_u)],
  \end{align}
  where $\HH(S^H,\PP^H)$ is the collection of predictable processes $G$ such
  that $$\int_0^TG_udS^H_u$$ is a $\PP^H$-martingale.
\end{lemma}
\begin{proof}
  Let
  \begin{align*}
    \mathcal{A}(x)\triangleq\{h: \ \EE^H[h]\le x\}.
  \end{align*}
  In order to prove that $-H$ is the unique solution of the optimization
  problem \eqref{eq:u(x)}, we need to show that
  \begin{align*}
    \widehat{h} \triangleq x-\int_0^TH_udS^H_u
  \end{align*}
  is an element of $\mathcal{A}(x)$, and for any $h\in\mathcal{A}(x)$, the
  following inequality holds true
  \begin{align}\label{eq:max_ineq}
    \EE[U(h)]\le\EE[U(\widehat{h})].
  \end{align}
  The fact that $\widehat{h}$ is an element of $\mathcal{A}(x)$ follows from
  the martingale property of $\int HdS^H$ under $\mathbb{P}^H$, see Definition
  \ref{def:price_proc}.

  Further, let $V(y)$ be the Legendre transform of $U(z)$, i.e.
  \begin{align}\label{eq:legendre}
    V(y)\triangleq\sup_{z\in\Real}\{U(z)-zy\},\quad y>0
  \end{align}
  It follows from \eqref{eq:legendre} that for any $y\ge 0$ and $z\in\Real$
  \begin{align*}
    U(z)\le V(y)+zy,
  \end{align*}
  and therefore
  \begin{align}\label{eq:arb_h}
    \EE[U(h)] & \le \EE\left[V\left(y\frac{d\PP^H}{d\PP}\right)\right] +
    \EE\left[hy\frac{d\PP^H}{d\PP}\right]
    = \EE\left[V\left(y\frac{d\PP^H}{d\PP}\right)\right] + \EE^H[h]y\nonumber \\
    &\le \EE\left[V\left(y\frac{d\PP^H}{d\PP}\right)\right] + xy,
  \end{align}
  where the last inequality follows by virtue of $h$ being an element of
  $\mathcal{A}(x)$.  On the other hand, the identity
  \begin{align*}
    U(I(y))=V(y)+yI(y), \quad \hbox{ where } I(y)=(U')^{-1}(y),\ \ y>0
  \end{align*}
  along with 
  \begin{align*}
    I\left(y\frac{d\PP^H}{d\PP}\right) =
    (U')^{-1}\left(y\frac{d\PP^H}{d\PP}\right) = \widehat{h},\ \
    y=\EE[U'(\widehat{h})]
  \end{align*}
  implies that
  \begin{align}\label{eq:h_hat}
    \EE[U(\widehat{h})]& =\EE\left[V\left(y\frac{d\PP^H}{d\PP}\right)\right] +
    \EE\left[\widehat{h}y\frac{d\PP^H}{d\PP}\right] =
    \EE\left[V\left(y\frac{d\PP^H}{d\PP}\right)\right]
    + \EE^H[\widehat{h}]y\nonumber\\
    &= \EE\left[V\left(y\frac{d\PP^H}{d\PP}\right)\right] + xy.
  \end{align}
  Now we compare \eqref{eq:arb_h} with \eqref{eq:h_hat} and conclude that
  \eqref{eq:max_ineq} holds true, and therefore
  \begin{align*}
    u(x)=\EE[U(\widehat{h})]=\EE[U(x-\int_0^TH_udS^H_u)].
  \end{align*}
\end{proof}

In what follows we are interested to find the answers to the following
question: Is it possible for a large trader to replicate another {\it non-traded}
  contingent claim $g$, that is, form a demand $H$ such that for some initial
  wealth $p$
  \begin{displaymath}
    p+ \int_0^T H_u dS_u^H = g? 
  \end{displaymath}

There are another two important questions to ask:
\begin{itemize}
  \item Does the price process $S^H$ exist for an arbitrary demand $H$?
  \item Provided that $S^H$ exists, is it unique?
\end{itemize}
The answers to the two latter questions are given in the companion paper by \cite{German:10b}, while in this paper we will be concerned with the former question.

\section{Replication}
\label{existence_and_uniqueness_rep_strat}

Consider an $\FF_T$-measurable random variable $g$ (alternatively we
will call it a \emph{non-traded} European-type contingent claim). The
utmost important question is what is the ``fair'' price of this
claim. We remind the reader, that in the framework of complete
financial model for a small economic agent the arbitrage-free price of
$g$ is given by
\begin{displaymath}
  p=\EE_{\PP^*}[g]
\end{displaymath}
and the unique hedging strategy can be found from the martingale
representation
\begin{displaymath}
  g=p+\int_0^TH_tdS_t, 
\end{displaymath}
when $S$ is a martingale under the unique martingale measure $\PP^*$.

The classical theory of asset pricing hinges on the crucial assumption that
the price per share of an asset does not depend on the size of the trade at
any moment in time. Moreover, when pricing by replication, it is understood
that the integrand and the integrator of the wealth process are not functions
of each other. In our large trader model the price process {\it is} a
(non-linear) function of demand. Therefore the problem of replication in the
illiquid market cannot be solved using the tools of the classical asset
pricing theory. In order to construct a perfect hedge for a non-traded contingent
claim it has to be taken into account that there is a back-and-forth
relationship between the size of the trade and the current price of the traded
asset. More precisely, we have the following definition.

\begin{definition}\label{def:hedging_strategy}
  Let $g,f=(f^j)_{1\le j\le J}$ be $\FF_T$-measurable random variables. A
  predictable process $H$ is called \emph{a hedging strategy of $g$}, if there
  exist $p\in\Real$, and a price process $S^H$ of $f$ under demand $H$ such
  that
  \begin{align*}
    g=p+\int_0^TH_udS^H_u.
  \end{align*}
\end{definition}

\begin{remark}
  Similarly to pricing by replication in the classical framework of a small
  economic agent, we will call $p$ a \emph{price of $g$}, since $p$ is the
  initial capital required for the perfect replication of~$g$. Note that it is
  not clear \emph{a priori} that $p$ is defined uniquely. We shall show below
  that the uniqueness always holds true for exponential utilities.
\end{remark}

\begin{remark}
  The above definition looks similar to the classical definition of a hedging
  strategy, with the crucial difference that the price process and the hedging
  strategy depend on each other.
\end{remark}

\begin{theorem}[Necessary condition]\label{thm:hedging_necessary_cond}
  Let $x\in\Real$ be the initial capital of the \mm. Assume that $f$
  and $g$ satisfy Assumption \ref{as:exp_moments}, and $U$ satisfies
  Assumption \ref{as:exp_utility}. Suppose there exists a hedging
  strategy $H$ of the contingent claim $g$ with price $p$. Then the
  unique pricing measure $\PP^H$ is given by the density
  \begin{equation}\label{eq:def_g_P}
    \frac{d\PP^H}{d\PP}=\frac{U'(x+p-g)}{\EE[U'(x+p-g)]},
  \end{equation}
  the price process is unique and is given by
  \begin{equation*}
    S^H_t = \EE^H[f|\FF_t].
  \end{equation*}
  Moreover, the following integral representation holds true
  \begin{equation}\label{eq:g_H_integr}
    \EE^H[g|\FF_t]=p+\int_0^tH_udS^H_u, \hbox{ for any } t\in[0,T].
  \end{equation}
\end{theorem}

\begin{remark}\label{remark:unique} From integral representation
  \eqref{eq:g_H_integr} we deduce that the hedging process $H$ is
  defined uniquely a.s. on $\Omega \times [0,T]$ with respect to
  $d\PP[\omega]\times d\variance{S^H}_t$ in the following sense. Due
  to the Assumption \ref{as:exp_moments}, $g$ has finite exponential
  moments and therefore it is also square-integrable. Therefore if
  there exists another hedging strategy $\widetilde H$ such that
  \begin{equation*}
    \EE^H[g|\FF_t]=p+\int_0^t\widetilde{H}_udS^{\widetilde{H}}_u,
  \end{equation*}
  then
  \begin{equation*}
    \sum_{j=1}^J\int_0^T|\widetilde{H}_t^j-H_t^j|^2d\variance{{S^H}^j}_t=0 \ \hbox{ a.s.}
  \end{equation*}
\end{remark}

\begin{proof}
  Since $H$ is a hedging strategy of the contingent claim $g$, there exists
  $S^H$, the price process of $f$ under demand $H$ along with the
  corresponding pricing measure $\PP^H$, such that
  \begin{align}\label{eq:g_necessary}
    g=p+\int_0^TH_udS^H_u.
  \end{align}
 
  Therefore by Definition \ref{def:price_proc},
  \begin{equation}
    \label{eq:g_P}
    \frac{d\PP^H}{d\PP} = \frac{U'(x-\int_0^TH_udS^H_u)}{\EE[U'(x-\int_0^TH_udS^H_u)]} = \frac{U'(x+p-g)}{\EE[U'(x+p-g)]}.
  \end{equation}

  Due to Assumption \ref{as:exp_utility}, the first derivative of the utility
  function $U$ is bounded from below and above by exponential
  functions. Therefore Assumption \ref{as:exp_moments} implies that the random
  variable \eqref{eq:g_P} (the density of the pricing measure $\PP^H$) is well
  defined. It is also unique, and so is the price process $S^H$,
  which by Definition \ref{def:price_proc} is equal to
  \begin{align*}
    S^H_t=\EE^H[f|\FF_t], \quad 0\le t\le T.
  \end{align*}
  Since $\int_0^tH_udS^H_u$ is a $\PP^H$-martingale, by applying conditional
  expectation to the both sides of \eqref{eq:g_necessary} we obtain
  \begin{align*}
    \EE^H[g|\FF_t]=p+\int_0^tH_udS^H_u, \hbox{ for any } t\in[0,T],
  \end{align*}
  along with
  \begin{align*}
    \EE^H[g]=p.
  \end{align*}
\end{proof}

We start the study of the existence of replication strategy with the following

\begin{lemma}\label{lemma:p_existence}
  Assume that $g$ satisfies Assumption \ref{as:exp_moments}, $U$ satisfies
  Assumption \ref{as:exp_utility}, and \eqref{eq:5} holds true.  Then for any
  $x\in\Real$ the equation
  \begin{equation}
    \label{eq:exist_p}
    \EE[(p-g)U'(x+p-g)] = 0    
  \end{equation}
  has a solution. If, in addition, the utility function $U$ is exponential,
  that is,
  \begin{equation*}
    U(x) = -\frac1\gamma e^{-\gamma x}
  \end{equation*}
  for some $\gamma>0$, than the solution of \eqref{eq:exist_p} is unique and
  given by
  \begin{equation*}
    \label{eq:exp_util_p}
    p = \frac{\mathbb{E}[g e^{-\gamma g}]}{\mathbb{E}[e^{-\gamma g}]}.   
  \end{equation*} 
\end{lemma}
\begin{proof}

  Let us consider the function $\alpha(p):\Real\rightarrow\Real$
  defined as
  \begin{equation*}
    \alpha(p)\triangleq\EE\left[(p-g)\frac{U'(x+p-g)}{U'(x+p)}\right].
  \end{equation*}
  We will show that $\alpha(p)$ has a zero (at least one). Since $U$ is a utility function (strictly increasing, strictly concave),
  $U'(x+p)$ is non-random and strictly positive for all $x$ and $p$. Therefore
  the existence of a solution of the equation
  \begin{equation}\label{eq:alpha-zero}
    \alpha(p)=0
  \end{equation}
  will imply the existence of a solution of \eqref{eq:exist_p}.

  In order to show the existence of a solution of
  \eqref{eq:alpha-zero} it is sufficient to show that
  \begin{align}
    &\lim_{p\rightarrow -\infty}\alpha(p)\le 0,\hbox{ and}\label{eq:alpha<0}\\
    &\lim_{p\rightarrow\infty}\alpha(p)\ge 0.\label{eq:alpha>0}
  \end{align}

  It follows directly from Assumption \ref{as:exp_utility} that for
  any $z,y\in\Real$
  \begin{equation*}
    e^{-c_2y}\le\frac{U'(z+y)}{U'(z)}\le e^{-c_1y}.
  \end{equation*}
  Therefore
  \begin{equation*}
    e^{c_2g}\le\frac{U'(x+p-g)}{U'(x+p)}\le e^{c_1g},
  \end{equation*}
  and due to Assumption \ref{as:exp_moments}, the random variable 
  $\frac{U'(x+p-g)}{U'(x+p)}$ has a finite positive expectation and is square-integrable. Hence by H\"older's inequality
  \begin{equation*}
    \EE\left[g\frac{U'(x+p-g)}{U'(x+p)}\right]<\infty,
  \end{equation*}
and therefore 
  \begin{align}
    \lim_{p\rightarrow-\infty}\alpha(p) = &\lim_{p\rightarrow-\infty} \EE\left[(p-g)\frac{U'(x+p-g)}{U'(x+p)}\right] = -\infty,\nonumber\\
    &\lim_{p\rightarrow\infty}\EE\left[p\frac{U'(x+p-g)}{U'(x+p)}\right] = \infty,\label{eq:alpha1}\\
    &\lim_{p\rightarrow\infty}\EE\left[g\frac{U'(x+p-g)}{U'(x+p)}\right]
    =C<\infty,\label{eq:alpha2}
  \end{align}
  for some constant $C\in\Real$.  As a direct consequence of
  \eqref{eq:alpha1} and \eqref{eq:alpha2} we obtain
  \begin{align*}
    \lim_{p\rightarrow\infty}\alpha(p) = \lim_{p\rightarrow\infty}\EE\left[(p-g)\frac{U'(x+p-g)}{U'(x+p)}\right]=\infty.
  \end{align*}
Therefore we conclude that \eqref{eq:alpha-zero} has a solution and consequently \eqref{eq:exist_p} has a solution.

Equation \eqref{eq:exist_p} can be written as
\begin{equation*}
  p=\frac{\EE[gU'(x+p-g)]}{\EE[U'(x+p-g)]},
\end{equation*}
and for an exponential utility function we have
\begin{equation*}
  p=\frac{\EE[g\exp\{-\gamma(x+p-g)\}]}{\EE[\exp\{-\gamma(x+p-g)\}]} = \frac{\EE[ge^{\gamma g}]}{\EE[e^{\gamma g}]},
\end{equation*}
which is the unique solution of \eqref{eq:exist_p}.
\end{proof}

\begin{remark}
  In the following theorem we will need the notion of
  \emph{completeness}. By completeness we will understand that every
  European-type derivative security (an $\FF_T$-measurable random
  variable) can be represented as a constant plus an integral with
  respect to the underlying security.
\end{remark}

  We shall now state

\begin{theorem}[Sufficient condition]\label{thm:sufficient_condition}
  Let $x\in\Real$ be the initial capital of the \mm. Assume that $f$ and $g$
  satisfy Assumption \ref{as:exp_moments}, $U$ satisfies Assumption
  \ref{as:exp_utility}, and \eqref{eq:5} holds true.  Let $\Ptilde$ be a
  probability measure such that for some $p\in\Real$
  \begin{equation}
    \frac{d\Ptilde}{d\PP}\triangleq\frac{U'(x+p-g)}{\EE[U'(x+p-g)]},
    \label{eq:ptilde_density}
  \end{equation}
  and
  \begin{equation}
    \Etilde[g] =p. \label{eq:ptilde_expect_g}
  \end{equation}
  Let us denote
  \begin{align}
    \label{eq:stilde}
    \Stilde_t = \Etilde[f|\FF_t],
  \end{align}
  and suppose that the model is \emph{complete} with respect to
  $\Stilde$. Then there exists a hedging strategy $H$ such that
  \begin{align*}
    g=p+\int_0^TH_udS^H_u,
  \end{align*}
  with $S^H=\Stilde$, the price process of $f$ under demand $H$.
\end{theorem}

\begin{remark}
  The constant $p$ allowing for the existence of pricing measure
  $\widetilde{\mathbb{P}}$ satisfying \eqref{eq:ptilde_density} and
  \eqref{eq:ptilde_expect_g} is exactly the one solving
  \eqref{eq:exist_p}. Note that, by Lemma \ref{lemma:p_existence}, $p$ is
  defined uniquely for an important class of exponential utilities.
\end{remark}

\begin{proof}
  First we observe that since the market is complete with respect to
  $\Stilde$, the probability measure $\Ptilde$ defined by
  \eqref{eq:ptilde_density} is the unique martingale measure of $\Stilde$.

  Since the market is complete with respect to $\Stilde$, any
  $\FF_T$-measurable random variable can be represented as a constant
  plus an integral with respect to $\Stilde$. Moreover, due to
  Assumption \ref{as:exp_moments}, random variable $g$ is
  square-integrable, and therefore there exists a predictable process
  $H$, such that the process $\int_0^tH_ud\Stilde_u$ is a
  $\Ptilde$-martingale, and
  \begin{align}\label{eq:g_repr}
    g=\Etilde[g]+\int_0^TH_ud\Stilde_u=p+\int_0^TH_ud\Stilde_u,
  \end{align}
  where the last equality follows from \eqref{eq:ptilde_expect_g}. Hence
  \begin{align*}
    \frac{d\Ptilde}{d\PP}=\frac{U'(x+p-g)}{\EE[U'(x+p-g)]}
    =\frac{U'(x-\int_0^TH_ud\Stilde_u)}{\EE[U'(x-\int_0^TH_ud\Stilde_u)]}.
  \end{align*}
  By Definition \ref{def:price_proc}, the above implies that $\PP^H=\Ptilde$
  is the unique pricing measure of $f$ under demand $H$, and $S^H=\Stilde$ is
  the unique price process of $f$ under demand $H$. Hence, by Definition
  \ref{def:hedging_strategy} $H$ is the hedging strategy of $g$.
\end{proof}

\begin{remark} We point out again that the integral representation
  \eqref{eq:g_repr} implies that the hedging process $H$ is defined uniquely
  a.s. on $\Omega \times [0,T]$ with respect to $d\PP[\omega]\times
  d\variance{S^H}_t$ in the sense explained in Remark \ref{remark:unique}. 
\end{remark}

The assumption on market completeness is essential for the sufficient
condition. Here we present an example of an incomplete market for which the
above theorem does not hold, and there is no hedging strategy for a particular
contingent claim $g$.

\begin{example} \label{ex:incompl} Consider the case of a one-dimensional
  Brownian Motion $B$. Let $g=B_T$, and let $U(x)=-e^{-x}$. Therefore
  \begin{align*}
    \ds\frac{d\Ptilde}{d\PP}=\frac{e^{-B_T}}{e^{\half T}}=e^{-B_T-\half T},
  \end{align*}
  and by Girsanov Theorem under the probability measure $\Ptilde$ there exists
  another Brownian Motion $\Btilde$ such that
  \begin{align*}
    \Btilde_t=B_t+t.
  \end{align*}
  Further, let
  \begin{align*}
    f=B_T-B_{\tau}+T-\tau=\int_0^T\ind_{[\tau,T]}(u)d\Btilde_u
  \end{align*}
  for some discrete time $0<\tau< T$. In this case
  \begin{align*}
    \Stilde_t&=\EE[\int_0^T\ind_{[\tau,T]}(u)d\Btilde_u|\FF_t]\\
    &=\int_0^t\ind_{[\tau,T]}(u)d\Btilde_u = \int_0^t\sigma_ud\Btilde_u, \quad
    \sigma_t=\ind_{[\tau,T]}(t),
  \end{align*}
  which implies that the the process $\Stilde$ is identically zero on the time
  interval $[0,\tau)$. There exists a martingale representation of
  \begin{align*}
    g=-B_T=-T+\int_0^Td\Btilde_u
  \end{align*}
  as an integral with respect to the Brownian Motion $\Btilde$. However, it is
  impossible to represent $g$ as an integral with respect to $\Stilde$, since
  the volatility $\sigma$ is zero on $[0,\tau]$. In other words, the model is
  incomplete with respect to $\Stilde$.
  \begin{flushright}$\square$\end{flushright}
\end{example}

\section{Completeness with respect to $\Stilde$.}
\label{sec:compl-with-resp}

As it was highlighted by Example \ref{ex:incompl}, for the existence of a
replication strategy it is necessary to verify that the market driven by the
price process $\Stilde$ is complete. This problem will be the focus of the
current section.

We start by recalling (without the proof) the following (well-known) fact lying in the
intersection of the Girsanov and the Martingale Representation theorems.

\begin{lemma}\label{lemma:compl_tilde_B}
  Consider a filtered probability space $(\Omega, \FF, (\FF_t)_{0\leq t\leq
    T}, \PP)$,

  \noindent where the filtration $(\FF_t)_{0\leq t\leq T}$ is generated by a
  $J$-dimensional Brownian motion $B$. Let $\Ptilde\sim\PP$ and let $\alpha =
  (\alpha_t)$ be a $J$-dimensional stochastic process such that
  \begin{equation}
    \label{eq:tilde_B}
    \Btilde_t \triangleq B_t + \int_0^t \alpha_u du, \quad 0\leq t\leq T, 
  \end{equation}
  is the $J$-dimensional Brownian Motion under $\Ptilde$. Then any
  $\Ptilde$-martingale $\Mtilde$ is a stochastic integral with respect
  to $\Btilde$.
\end{lemma}

\begin{proposition}
  \label{prop:compl}
  Assume that $f$ and $g$ satisfy Assumption \ref{as:exp_moments}, $U$
  satisfies Assumption \ref{as:exp_utility} and that the filtration is
  generated by the Brownian motion $B$ (that is, \eqref{eq:5} holds true).
  Let $\Ptilde$ be the probability measure defined in \eqref{eq:ptilde_density},
  $\Stilde$ be the price process defined in \eqref{eq:stilde} and $\Btilde$ be
  the $J$-dimensional Brownian motion under $\Ptilde$ given by
  \eqref{eq:tilde_B}.

  Then the financial model determined by the price process $\Stilde$ given by
  \eqref{eq:stilde} is complete if and only if the $J\times J$-dimensional
  matrix process $\widetilde{\sigma} = (\widetilde{\sigma}_t)$ in the
  martingale representation
  \begin{equation}
    \label{eq:stilde_sigmatilde}
    \Stilde_t=\Stilde_0 + \int_0^t \sigt_u\cdot d\Btilde_u 
  \end{equation}
  or, component-wise,
  \begin{displaymath}
    \Stilde^i_t=\Stilde^i_0 + \sum_{1\leq j\leq J} \int_0^t
    \sigt_u^{ij}d\Btilde^j_u , \quad 1\leq i\leq J, 
  \end{displaymath}
  has full rank almost everywhere with respect to the product measure
  $d\mathbb{P}[\omega]\times dt$.
\end{proposition}

\begin{proof}
  Let us assume that the matrix $\sigt$ has full rank (or in other words it is
  invertible) almost everywhere with respect to the product measure
  $d\PP[\omega]\times dt$. Let $\Gamma$ be an $\FF_T$-measurable
  $J$-dimensional random variable. Then by Lemma \ref{lemma:compl_tilde_B}
  there exists an adapted $J\times J$-dimensional matrix process
  $\gamma=(\gamma_t)_{0\le t\le T}$ such that
  \begin{align*}
    \Gamma = \Etilde[\Gamma]+\int_0^T\gamma_u\cdot d\Btilde_u.
  \end{align*}
  Since $\sigt$ is invertible $d\PP[\omega]\times dt$-almost everywhere, and
  by \eqref{eq:stilde_sigmatilde}, the above expression is equal to
  \begin{align*}
    \Gamma = \Etilde[\Gamma]+\int_0^T\gamma_u\cdot \sigt_u^{-1}\cdot
    d\Stilde_u.
  \end{align*}
  Hence we conclude that the market is complete with respect to $\Stilde$,
  since every contingent claim is replicable.

  The proof of the converse statement will be done by contradiction.  Assume
  that the market is complete with respect to $\Stilde$, and suppose that the
  matrix $\sigt_t$ is not invertible $d\PP[\omega]\times dt$-almost
  everywhere. Let us define the following adapted vector process
  $a=(a^j_t)_{0\le t\le T}, \ 1\le j\le J$. For each $0\le t\le T$, let the
  vector $a_t$ be equal to a non-trivial vector from the null-space of the
  matrix $\sigt_t$. Since we assumed that $\sigt_t$ is not invertible, the
  null-space of $\sigt_t$ is not trivial. That is
  \begin{align*}
    a_t=n_t, \ n_t\in\hbox{Null}(\sigt_t), \ n_t\ne \vec{0},
  \end{align*}
  where $\hbox{Null}(\sigt_t)$ is the Null-space of $\sigt_t$, and $\vec{0}$
  is a $J$-dimensional zero vector. Let us consider two $J$-dimensional
  $\Ptilde$-martingales
  \begin{align*}
    \int_0^ta_ud\Btilde_u, \hbox{ and } \int_0^t\sigt_u\cdot d\Btilde_u,
  \end{align*}
  or component-wise,
  \begin{align*}
    \int_0^ta_u^kd\Btilde^k_u, \hbox{ and } \sum_{1\le j\le
      J}\int_0^t\sigt_u^{kj} d\Btilde_u^j,\quad 1\le k\le J.
  \end{align*}
  Due to the choice of $a$, at any time $t$ the dot-product of $a_t$ and
  $\sigt_t^{\mathrm{T}}$ is a zero vector, and therefore their cross-variation
  is
  \begin{align*}
    \variance{\int_0^{\cdot}a_ud\Btilde_u,\ \int_0^{\cdot}\sigt_u\cdot
      d\Btilde_u}_t = \int_0^ta_u\cdot\sigt_u^{\mathrm{T}}du = \vec{0}, \quad
    0\le t\le T,
  \end{align*}
  as well as
  \begin{align*}
    \variance{\int_0^{\cdot}a_ud\Btilde_u,\ \Stilde_{\cdot}}_t=\vec{0}.
  \end{align*}
  Since the cross-variation process of two $\Ptilde$-martingales is zero, the
  product of these martingales is a zero process, which implies that
  $\int_0^ta_ud\Btilde_u$ and $\Stilde_t$ are orthogonal. Notice that the
  random variable
  \begin{align}\label{eq:non_repl}
    \int_0^Ta_ud\Btilde_u
  \end{align}
  is different from zero because by the assumption the matrix $\sigt$ does not
  have full rank almost everywhere with respect to the product measure
  $d\PP[\omega]\times dt$. We assumed that the market is complete and
  therefore each non-trivial $\FF_T$-measurable random variable is
  non-trivially replicable. However, due to orthogonality of
  $\int_0^ta_ud\Btilde_u$ and $\Stilde_t$, the $\FF_T$-measurable random
  variable \eqref{eq:non_repl} is not $\Stilde$-replicable, which leads us to
  a contradiction. Therefore the market is incomplete with respect to
  $\Stilde$.
\end{proof}

The computation of the $J\times J$-matrix volatility process
$\widetilde{\sigma}=(\widetilde{\sigma}_t)$ can often be done with the help of
the Clark-Ocone formula and the Malliavin calculus as is illustrated in the
following lemma. For a random variable $\psi$ we denote by $\Dm_t(\psi) =
(\Dm^j_t( \psi))_{1\leq j\leq J}$ the Malliavin derivative of $\psi$ at time
$t$ with respect to the Brownian Motion $B$.

\begin{lemma}\label{lemma:malliavin}
  In addition to conditions of Proposition \ref{prop:compl} assume that $f$
  and $g$ satisfy Assumption \ref{as:exp_moments} and are Malliavin
  differentiable.
 
  Then the matrix $\widetilde{\sigma}_t$ in the integral representation
  \eqref{eq:stilde_sigmatilde} is given by
  \begin{equation}
    \label{eq:sigma_ij}
    \sigt^{ij}_t= \Etilde[\Dm^j_t(f^i) + A(x+p-g) (f^i-\Stilde_t^i) 
    \Dm_t^j(g)|\FF_t],\quad 1\leq i,j\leq J,    
  \end{equation}
  where $A= A(x)$ is the absolute risk-aversion coefficient of $U$ given by
  \begin{equation}
    \label{eq:11}
    A(x) \triangleq -\frac{U''(x)}{U'(x)}, \quad x\in \mathcal{R}.
  \end{equation}
\end{lemma}

\begin{proof}
  Our goal is to compute the integrand in the form of a $J\times
  J$-dimensional matrix process in the martingale representation
  \eqref{eq:stilde_sigmatilde} of the process $\Stilde$.

  Since the density process is
  $Z_t=\EE\left[\left.\frac{d\Ptilde}{d\PP}\right|\FF_t\right]$, it can be
  also expressed as
  \begin{align}\label{eq:Zt_g}
    Z_t=\EE\left[\left.\frac{d\Ptilde }{d\PP
        }\right|\FF_t\right]=\frac{\EE[U'(x+p-g) |\FF_t] }{\EE[U'(x+p-g)]}.
  \end{align}
  Let us introduce another $J$-dimensional $\PP$-martingale. Let
  \begin{align*}
    R_t=\EE[Z_T\Stilde_T|\FF_t].
  \end{align*}
  The density process is an exponential martingale and so for some adapted
  $J$-dimensional process $\alpha=(\alpha_t)_{0\le t\le T}$ we have
  $Z_t=\mathcal{E}(\alpha_t\!\cdot\! B_t)$, where $\mathcal{E}$ is the
  stochastic exponent. Since
  \begin{align*}
    \Stilde_t=\Etilde[f|\FF_t]=\frac{1}{Z_t}\EE[Z_Tf|\FF_t],
  \end{align*}
  it follows that $R_t=Z_t\Stilde_t$. The process $R$ is a $\PP$-martingale,
  and since the density process $Z$ is strictly positive $\PP$-almost surely,
  it has the martingale representation
  \begin{align}
    R_t=\Stilde_0+\int_0^tZ_u {\Sigma}_u\!\cdot\!  dB_u\label{eq:R_mart}
  \end{align}
  for some progressively measurable matrix process ${\Sigma}$. By
  differentiating $Z_t\Stilde_t$ we obtain that the differential of $R$ is
  \begin{align}
    dR_t=Z_t\widetilde{\sigma}_t\!\cdot\!dB_t +
    Z_t\Stilde_t\!\cdot\!\alpha_t\!\cdot\!  dB_t.\label{eq:R2}
  \end{align}
  By comparing (\ref{eq:R2}) and the differential of (\ref{eq:R_mart}) we
  obtain
  \begin{align}
    \widetilde{\sigma}_t=\Sigma_t-\Stilde_t\!\cdot\!\alpha_t.\label{eq:sigma}
  \end{align}
  Notice that the above expression for $\widetilde{\sigma}$ has two
  processes $\Sigma$ and $\alpha$ that are only known to exist and to
  be unique (in the sense discussed in Remark \ref{remark:unique}),
  but are not known explicitly.

  Due to the assumption that the random variable $g$ has exponential moments
  and is Malliavin differentiable and using the fact that
  $dZ_t=Z_t\alpha_t\!\cdot\! dB_t$ we deduce from Clark-Ocone formula that
  \begin{align}
    Z_t\alpha_t^j=\EE[\Dm_t^j(Z_T) |\FF_t],\quad 1\le j\le
    J. \label{eq:Zalpha}
  \end{align}
  Similarly (\ref{eq:R_mart}) implies that
  \begin{align}
    Z_t\Sigma_t^{ij}=\EE[\Dm_t^j(Z_Tf^i)
    |\FF_t]=\EE[Z_T\Dm_t^j(f^i)+f^i\Dm_t^j(Z_T)|\FF_t],\nonumber
  \end{align}
  and therefore
  \begin{align}
    \Sigma_t^{ij}=\Etilde[\Dm_t^j(f^i)|\FF_t]+\EE\left[\left.\frac{1}{Z_t}f^i
        \Dm_t^j(Z_T)\right|\FF_t\right].\label{eq:Sigma}
  \end{align}
  Hence by combining (\ref{eq:sigma}), (\ref{eq:Zalpha}), (\ref{eq:Sigma}) we
  obtain
  \begin{align}\label{eq:sigma2}
    \sigt^{ij}_t& =\Etilde[\Dm_t^j(f^i)|\FF_t] +
    \EE\left[\left.\frac{1}{Z_t}f^i\Dm_t^j(Z_T)\right|\FF_t\right]
    -\Stilde_t^i\EE\left[\left. \frac{1}{Z_t}\Dm_t^j(Z_T)\right|\FF_t
    \right]\nonumber\\
    &=\Etilde[\Dm_t^j(f^i)|\FF_t] +\EE\left[\left. \frac{1}{Z_t}(f^i -
        \Stilde_t^i)\Dm_t^j(Z_T)\right|\FF_t\right].
  \end{align}

  Now let us use the assumption that the random variable $g$ is Malliavin
  differentiable. Then from \ref{eq:Zt_g} it follows that
  \begin{align*}
    \Dm_t^jZ_T = \frac{-U''(x+p-g)\Dm_t^j(g)}{\EE[U'(x+p-g)]} =
    \frac{A(x+p-g)U'(x+p-g)\Dm_t^j(g)}{\EE[U'(x+p-g)]},
  \end{align*}
  where $A(x)$ is the absolute risk-aversion coefficient as defined in
  \eqref{eq:11}. Therefore the last term of \eqref{eq:sigma2} can be expressed
  as
  \begin{align*}
    \EE\left[\left. \frac{1}{Z_t}(f^i -
        \Stilde_t^i)\Dm_t^j(Z_T)\right|\FF_t\right] = \Etilde[ A(x+p-g)
    (f^i-\Stilde_t^i) \Dm_t^j(g)|\FF_t].
  \end{align*}
  We can now put everything together to obtain
  \begin{align*}
    \sigt^{ij}_t= \Etilde[\Dm^j_t(f^i) + A(x+p-g) (f^i-\Stilde_t^i)
    \Dm_t^j(g)|\FF_t].
  \end{align*}
\end{proof}

The expression \eqref{eq:sigma_ij} for $\widetilde{\sigma}$ can be simplified
if we assume that $U$ is the exponential utility function, and that the
contingent claim $g$ is a \emph{standard} European option on $f$. In the
formulation of the result we shall use the notation
\begin{equation}
  \label{eq:13}
  \covtilde(\alpha,\beta|\FF_t) = \Etilde[\alpha \beta|\FF_t] -
  \Etilde[\alpha|\FF_t]\Etilde[\beta|\FF_t] 
\end{equation}
for the conditional covariance of the random variables $\alpha$ and $\beta$
with respect to the pricing measure $\widetilde{\mathbb{P}}$ and the
information at time $t$.

\begin{lemma}
  \label{lem:G(f)}
  In addition to conditions of Lemma \ref{lemma:malliavin} assume that
  \begin{displaymath}
    g=G(f)
  \end{displaymath}
  for some almost everywhere differentiable function $G:
  \Real^J\longrightarrow\Real$, and that the utility function $U=U(x)$ is of
  exponential form:
  \begin{displaymath}
    U(x) = \frac1\gamma e^{-\gamma x}, \quad x\in \mathcal{R},
  \end{displaymath}
  for some $\gamma >0$. Then the matrix $\widetilde{\sigma}_t$ in the integral
  representation \eqref{eq:stilde_sigmatilde} is given by
  \begin{align}
    \widetilde{\sigma}^{ij}_t&=\Etilde[\Dm_t^j(f^i) + \gamma(f^i -
    \widetilde{S}^i_t) \sum_{1\leq k \leq J}\Dm^j_t(f^k)
    \frac{\partial}{\partial x^k} G(f)|\FF_t] \label{eq:sigma_deriv}\\
    &=\Etilde[\Dm_t^j(f^i)|\FF_t] + \gamma \covtilde(f^i, \sum_{1\leq k \leq
      J}\Dm^j_t(f^k) \frac{\partial}{\partial x^k} G(f)|\FF_t],\label{eq:sigma_deriv_covar} \\
    & \quad 1\leq i,j\leq J.\nonumber
  \end{align}
\end{lemma}
\begin{proof}
  First we remind that for the exponential utility $U(x) = \frac1\gamma
  e^{-\gamma x}$, the absolute risk-aversion coefficient $A(x)$ is constant
  and is equal to
  \begin{align*}
    A(x)=\gamma,\hbox{ for every } x\in\Real.
  \end{align*}
  Equation \eqref{eq:sigma_deriv} follows directly from \eqref{eq:sigma_ij}
  and from the assumption that $g$ is an almost everywhere differentiable
  function of $f$,
  \begin{align*}
    \Dm_t^j(g)=\Dm_t^j(G(f))=\sum_{1\leq k \leq J}\Dm^j_t(f^k)
    \frac{\partial}{\partial x^k} G(f), \quad 1\leq j\leq J,
  \end{align*}
  which leads to
  \begin{align*}
    \sigt^{ij}_t&= \Etilde[\Dm^j_t(f^i) + A(x+p-g) (f^i-\Stilde_t^i)
    \Dm_t^j(g)|\FF_t] \\
    &= \Etilde[\Dm_t^j(f^i) + \gamma(f^i - \widetilde{S}^i_t) \sum_{1\leq k
      \leq J}\Dm^j_t(f^k) \frac{\partial}{\partial x^k} G(f)|\FF_t].
  \end{align*}
  As for \eqref{eq:sigma_deriv_covar}, we have
  \begin{align*}
    \sigt^{ij}_t&= \Etilde[\Dm^j_t(f^i) + A(x+p-g) (f^i-\Stilde_t^i)
    \Dm_t^j(g)|\FF_t]=\\
    &=\Etilde[\Dm_t^j(f^i)|\FF_t] + \gamma\Etilde[f^i\Dm_t^j(g)|\FF_t]
    -\gamma\Etilde[\Stilde_t^i\Dm_t^j(g)|\FF_t]\\
    &=\Etilde[\Dm_t^j(f^i)|\FF_t]+\gamma\Etilde[f^i\Dm_t^j(g)|\FF_t]
    - \gamma\Stilde_t^i\Etilde[\Dm_t^j(g)|\FF_t]\\
    &=\Etilde[\Dm_t^j(f^i)|\FF_t]+\gamma\Etilde[f^i\Dm_t^j(g)|\FF_t]
    - \gamma\Etilde[f^i|\FF_t]\Etilde[\Dm_t^j(g)|\FF_t]\\
    &=\Etilde[\Dm_t^j(f^i)|\FF_t]+\gamma\covtilde(f^i,\Dm_t^j(g)|\FF_t)\\
    &=\Etilde[\Dm_t^j(f^i)|\FF_t]+\gamma\covtilde(f^i,\sum_{1\leq k \leq
      J}\Dm^j_t(f^k) \frac{\partial}{\partial x^k} G(f)|\FF_t).
  \end{align*}
\end{proof}

The computation of the matrix $\widetilde{\sigma}$ becomes particularly simple
in the case when the payoffs of traded contingents claims $f$ are the terminal
values of the Brownian Motions $B$, that is,
\begin{equation}
  \label{eq:12}
  f^j = B_T^j, \quad 1\leq j\leq J. 
\end{equation}
Note that in this case, in the absence of any trading by the large investor,
that is, in the case $H=0$ the price process of the stocks is given by $B$. In
other words, in the absence of the large investor we have the
multi-dimensional Bachelier model. In the statement of the next lemma we use
the standard notation
\begin{displaymath}
  \delta_{ij} = 1_{\{i=j\}}
\end{displaymath}
for the Kronecker delta.

\begin{lemma}
  \label{lem:Bachelier}
  In addition to conditions of Lemma \ref{lem:G(f)} assume \eqref{eq:12}. Then
  \begin{equation} 
    \label{eq:sigma_G(f)_bach}
    \widetilde{\sigma}^{ij}_t = \delta_{ij} + \gamma \covtilde\left(\left.B_T^i,
    \frac{\partial}{\partial x^j} G(B_T)\right|\FF_t\right), \quad 
    1\leq i,j\leq J.
  \end{equation} 
\end{lemma}

\begin{proof}
  By the assumption that the payoffs of the contingent claims are Brownian
  Motions,
  \begin{align*}
    \Etilde[\Dm_t^j(f^i)|\FF_t]=\Etilde[\Dm_t^j(B_T^i)|\FF_t]=\delta_{ij}.
  \end{align*}
  It follows that
  \begin{align*}
    \sum_{1\leq k \leq J}\Dm^j_t(B_T^k) \frac{\partial}{\partial x^k}
    G(B_T)=\frac{\partial}{\partial x^j} G(B_T).
  \end{align*}
  Now we can put everything together to obtain
  \begin{align*}
    \sigt^{ij}&=\Etilde[\Dm_t^j(f^i)|\FF_t]+\gamma\covtilde\left(\left.f^i,\sum_{1\leq k
      \leq J}\Dm^j_t(f^k) \frac{\partial}{\partial x^k}
    G(B_T)\right|\FF_t\right)\\
    & = \delta_{ij} + \gamma \covtilde\left(\left.B_T^i, \frac{\partial}{\partial x^j}
    G(B_T)\right|\FF_t\right).
  \end{align*}
\end{proof}

As an important corollary we state the following result showing that in the
framework of the Bachelier model a large class of \emph{convex} (in
appropriate sense) contingent
claims $G(f)$ is replicable. For instance this includes a convex
combination of \emph{long} positions in European calls written on each
individual asset.

\begin{corollary}
  \label{corr:Bachelier}
  Assume conditions of Lemma \ref{lem:Bachelier} . Let $$G = G(x^1,\dots,x^J)=\sum_{j=1}^Jc_j\varphi_j(x^j),$$
where $c_j\ge 0$, $c_j\in\Real$ and $\varphi_j=\varphi_j(x)$ are
one-dimensional convex functions. Then for any contingent claim $g$ of the form
  \begin{equation}
    \label{eq:14}
    g = G(f) = G(B_T), 
  \end{equation} hedging strategy $H$ exists and is defined uniquely $d\mathbb{P}[\omega] \times dt$-a.s.
\end{corollary}

\begin{proof}
 Let us observe that by the assumption the function $G$ is
  convex in each $x^j$ direction, that is its partial
  derivatives $\frac{\partial G}{\partial
    x^j}=c_j\frac{d\varphi^j(x^j)}{dx}$ are non-decreasing
  functions in each $x^j$ direction. Hence for any
  $\mathbf{x}_1 = (x_1^1,\dots,x_1^J) \in \Real^J,$ $\mathbf{x}_2 =
  (x_2^1,\dots,x_2^J) \in \Real^J$,
  and for any $1\le i,j\le J$
  \begin{align*}
    (x_1^i-x_2^i)\left(\frac{\partial G(\mathbf{x}_1)}{\partial
        x^j}-\frac{\partial G(\mathbf{x}_2)}{\partial x^j}\right) 
    = (x_1^i-x_2^i)c_j\left(\frac{d\varphi^j(x^j_1)}{dx}-\frac{d\varphi^j(x^j_2)}{dx}\right) \ge 0,
  \end{align*}
  and therefore
  \begin{equation}
    \label{eq:comonotonicity}
    \begin{split}
      c_j\left(B_T^i(\omega_1)-B_T^i(\omega_2)\right)
      \left(\frac{d\varphi^j(B_T(\omega_1))}{dx}-\frac{d\varphi^j(B_T(\omega_1))}{dx}\right)\ge0, \\
      d\Ptilde[\omega_1]\otimes d\Ptilde[\omega_2]-a.s.
    \end{split}
  \end{equation}
  Condition \eqref{eq:comonotonicity} means that the random variables $B_T^i$ and
  $c_j\frac{d\varphi^j(B_T)}{dx}$ are \emph{co-monotone}. It follows that
  \begin{equation}
    \label{eq:aplic_comonot}
     c_j\!\!\int_{\Omega_1\times\Omega_2} \negthickspace\negthickspace\negthickspace\negthickspace\negthickspace
      \left(B_T^i(\omega_1)-B_T^i(\omega_2)\right)\!\!
      \left(\frac{d\varphi^j(B_T(\omega_1))}{dx}-\frac{d\varphi^j(B_T(\omega_1))}{dx}\right)\!
      d\Ptilde[\omega_1]\times d\Ptilde[\omega_2]\ge 0.
 \end{equation}
  By Fubini's theorem \eqref{eq:aplic_comonot} is equal to
  \begin{align*}
    &c_j\int_{\Omega_1}
    \left(B_T^i(\omega_1)-\int_{\Omega_2}B_T^i(\omega_2)d\Ptilde[\omega_2]
    \right)\\
    &\qquad\qquad\times
    \left(\frac{d\varphi^j(B_T(\omega_1))}{dx}-\int_{\Omega^2}\frac{d\varphi^j(B_T(\omega_1))}{dx}d\Ptilde[\omega_2]
    \right)d\Ptilde[\omega_1]\\
    &=\Etilde\left[(B_T^i-\Etilde[B_T^i])\left(c_j\frac{d\varphi^j(B_T)}{dx} - \Etilde\left[c_j\frac{d\varphi^j(B_T)}{dx}\right]\right)\right]\\
    &=\Etilde\left[(B_T^i-\Etilde[B_T^i])\left(\frac{\partial
          G(B_T)}{\partial x^j} - \Etilde\left[\frac{\partial G(B_T)}{\partial x^j}\right]\right)\right]
    =\covtilde(B_T^i,\frac{\partial G(B_T)}{\partial x^j}).
  \end{align*}
  Hence the unconditional covariance of $B_T^i$ and $\frac{\partial G(B_T)}{\partial x^j}$ is
  non-negative. By the similar argument one can show that the conditional
  covariance of $B_T^i$ and $\frac{\partial G(B_T)}{\partial x^j}$ is
  \begin{align*}
    \covtilde\left(\left.B_T^i,\frac{\partial G(B_T)}{\partial x^j}\right|\FF_t\right)\ge 0.
  \end{align*}
We now notice that because of the independence of individual Brownian
Motions $(B^i)_{1\le i\le J}$ of the $J$-dimensional Brownian Motion $B$
  \begin{align*}
    \sigt_t^{ij} = \delta_{ij}+\gamma\covtilde\left(\left.B_T^i,\frac{\partial G(B_T)}{\partial
          x^j}\right|\FF_t\right)= 0 \quad d\mathbb{P}[\omega] \times dt\hbox{-a.s., when } i\ne j,
  \end{align*}
and  
  \begin{align*}
    \sigt_t^{ij} = \delta_{ij}+\gamma\covtilde\left(\left.B_T^i,\frac{\partial G(B_T)}{\partial
          x^j}\right|\FF_t\right)\ge 1 \quad d\mathbb{P}[\omega] \times dt\hbox{-a.s., when } i=j.
  \end{align*}
Therefore the covariance matrix $\sigt$ is a diagonal matrix with
non-zero entries on the diagonal and hence invertible $d\mathbb{P}[\omega] \times dt$-a.s. 

  Finally we deduce from Theorem \ref{thm:sufficient_condition} and
  Proposition \ref{prop:compl} that a hedging strategy $H$ exists and is
  defined uniquely $d\mathbb{P}[\omega] \times dt$-a.s.
\end{proof}

\bibliographystyle{plainnat}
\bibliography{finance}    

\begin{thebibliography}{14}
\providecommand{\natexlab}[1]{#1}
\providecommand{\url}[1]{\texttt{#1}}
\expandafter\ifx\csname urlstyle\endcsname\relax
  \providecommand{\doi}[1]{doi: #1}\else
  \providecommand{\doi}{doi: \begingroup \urlstyle{rm}\Url}\fi

\bibitem[Amihud et~al.(2005)Amihud, Mendelson, and Pedersen]{AmihudMendPed:05}
Yakov Amihud, Haim Mendelson, and Lasse~Heje Pedersen.
\newblock Liquidity and asset prices.
\newblock \emph{Foundations and Trends in Finance}, 1\penalty0 (4):\penalty0
  369--364, 2005.

\bibitem[Bank and Baum(2004)]{BankBaum:04}
Peter Bank and Dietmar Baum.
\newblock Hedging and portfolio optimization in financial markets with a large
  trader.
\newblock \emph{Math. Finance}, 14\penalty0 (1):\penalty0 1--18, 2004.
\newblock ISSN 0960-1627.

\bibitem[{\c C}etin et~al.(2004){\c C}etin, Jarrow, and
  Protter]{CetinJarrowProtter:04}
Umut {\c C}etin, Robert Jarrow, and Philip Protter.
\newblock Liquidity risk and arbitrage pricing theory.
\newblock \emph{Finance and Stochastics}, 8:\penalty0 311--341, 2004.

\bibitem[Cvitani{\'c} and Karatzas(1995)]{CvitKar:95}
Jak{\v s}a Cvitani{\'c} and Ioannis Karatzas.
\newblock Hedging and portfolio optimization under transaction costs:
  martingale approach.
\newblock \emph{Mathematical Finance}, 6:\penalty0 133--165, 1995.

\bibitem[Cvitani{\'c} and Ma(1996)]{CvitMa:96}
Jak{\v{s}}a Cvitani{\'c} and Jin Ma.
\newblock Hedging options for a large investor and forward-backward {SDE}'s.
\newblock \emph{Ann. Appl. Probab.}, 6\penalty0 (2):\penalty0 370--398, 1996.
\newblock ISSN 1050-5164.

\bibitem[Duffie et~al.(2005)Duffie, G{\^a}rleanu, and Pedersen]{DufGarPed:05}
Darrell Duffie, Nicolae G{\^a}rleanu, and Lasse~Heje Pedersen.
\newblock Over-the-counter markets.
\newblock \emph{Econometrica}, 73\penalty0 (6):\penalty0 1815--1847, 2005.

\bibitem[Frey(1998)]{Frey:98}
R{\"u}diger Frey.
\newblock Perfect option hedging for a large trader.
\newblock \emph{Finance and Stochastics}, 2:\penalty0 115--141, 1998.

\bibitem[G{\^a}rleanu and Pedersen(2004)]{GarPed:04}
Nicolae G{\^a}rleanu and Lasse~Heje Pedersen.
\newblock Adverse selection and the required return.
\newblock \emph{Review of Financial Studies}, 17:\penalty0 643--665, 2004.

\bibitem[G{\^a}rleanu et~al.()G{\^a}rleanu, Pedersen, and
  Poteshman]{GarPedPot:06}
Nicolae G{\^a}rleanu, Lasse~Heje Pedersen, and Allen Poteshman.
\newblock Demand-based option pricing.
\newblock Forthcoming.
\newblock URL \url{http://faculty.haas.berkeley.edu/garleanu/DBOP.pdf}.

\bibitem[German(Forthcoming)]{German:10b}
David German.
\newblock Pricing in an equilibrium based model for a large investor.
\newblock Forthcoming.

\bibitem[Grossman and Miller(1988)]{GrossMiller:88}
Sanford Grossman and Merton~H. Miller.
\newblock Liquidity and market structure.
\newblock \emph{The Journal of Finance}, 43\penalty0 (3):\penalty0 617--633,
  1988.

\bibitem[Huang(2003)]{Huang:03}
Ming Huang.
\newblock Liquidity shocks and equilibrium liquidity premia.
\newblock \emph{Journal of Economic Theory}, 109:\penalty0 104--129, 2003.

\bibitem[Karatzas et~al.(1990)Karatzas, Lehoczky, and Shreve]{KarLehShreve:90}
Ioannis Karatzas, John~P. Lehoczky, and Steven~E. Shreve.
\newblock Existence and uniqueness of multi-agent equilibrium in a stochastic,
  dynamic consumption / investment model.
\newblock \emph{Mathematics of Operations Research}, 15\penalty0 (1):\penalty0
  80--128, 1990.

\bibitem[Stoll(1978)]{Stoll:78}
Hans~R. Stoll.
\newblock The supply of dealer services in securities markets.
\newblock \emph{The Journal of Finance}, 33\penalty0 (4):\penalty0 1133--1151,
  1978.

\end{thebibliography}
\end{document}